\documentclass[12pt]{article}

\usepackage[T1]{fontenc}
\usepackage{kpfonts}

\usepackage{setspace}
\usepackage[margin=1.25in]{geometry}

\usepackage[dvipsnames]{xcolor}

\usepackage{pdfpages}
\usepackage{float}
\usepackage{multirow}
\usepackage{caption}
\usepackage{subcaption}
\usepackage{epigraph}

\usepackage{mathtools}
\usepackage{amsthm}
\usepackage{aliascnt} 
\usepackage{accents}
\usepackage{dutchcal}
\usepackage{bbm}

\usepackage{enumitem}
\usepackage{sgame}

\usepackage{tikz}
\usepackage{tikz-cd}
\usetikzlibrary{calc,shapes,arrows}

\usepackage{tcolorbox}
\usepackage{listings}

\usepackage[normalem]{ulem}

\usepackage[round]{natbib}
\DeclareMathOperator{\inter}{int}

\DeclareMathOperator*{\argmax}{arg\,max}

\newtheorem{theorem}{Theorem}[section]

\newaliascnt{proposition}{theorem}
\newtheorem{proposition}[proposition]{Proposition}
\aliascntresetthe{proposition}

\newaliascnt{lemma}{theorem}

\aliascntresetthe{lemma}

\newaliascnt{corollary}{theorem}
\newtheorem{corollary}[corollary]{Corollary}
\aliascntresetthe{corollary}

\newaliascnt{claim}{theorem}

\aliascntresetthe{claim}

\theoremstyle{definition}

\newaliascnt{definition}{theorem}
\newtheorem{definition}[definition]{Definition}
\aliascntresetthe{definition}

\newaliascnt{example}{theorem}
\newtheorem{example}[example]{Example}
\aliascntresetthe{example}

\newaliascnt{assumption}{theorem}

\aliascntresetthe{assumption}

\newaliascnt{condition}{theorem}

\aliascntresetthe{condition}

\newaliascnt{question}{theorem}

\aliascntresetthe{question}

\newaliascnt{remark}{theorem}

\aliascntresetthe{remark}

\newaliascnt{remarks}{theorem}

\aliascntresetthe{remarks}

\newaliascnt{aside}{theorem}

\aliascntresetthe{aside}

\newaliascnt{note}{theorem}

\aliascntresetthe{note}

\usepackage{thmtools}
\usepackage{thm-restate}

\usepackage{hyperref}

\hypersetup{
    colorlinks=true,
    linkcolor=OrangeRed,
    filecolor=Thistle,
    urlcolor=Thistle,
    citecolor=Thistle,
}

\usepackage[nameinlink]{cleveref}

\crefname{theorem}{theorem}{theorems}
\Crefname{theorem}{Theorem}{Theorems}

\crefname{proposition}{proposition}{propositions}
\Crefname{proposition}{Proposition}{Propositions}

\crefname{lemma}{lemma}{lemmas}
\Crefname{lemma}{Lemma}{Lemmas}

\crefname{corollary}{corollary}{corollaries}
\Crefname{corollary}{Corollary}{Corollaries}

\crefname{claim}{claim}{claims}
\Crefname{claim}{Claim}{Claims}

\crefname{definition}{definition}{definitions}
\Crefname{definition}{Definition}{Definitions}

\crefname{example}{example}{examples}
\Crefname{example}{Example}{Examples}

\crefname{assumption}{assumption}{assumptions}
\Crefname{assumption}{Assumption}{Assumptions}

\makeatletter
\let\cref@old@isrefconsecutive\cref@isrefconsecutive

\def\cref@isrefconsecutive#1#2{%
  \begingroup
    \def\cref@assumptiontype{assumption}%
    \cref@gettype{#1}{\cref@typea}%
    \ifx\cref@typea\cref@assumptiontype
      \endgroup
      \@cref@refconsecutivefalse
    \else
      \endgroup
      \cref@old@isrefconsecutive{#1}{#2}%
    \fi
}
\makeatother

\crefname{condition}{condition}{conditions}
\Crefname{condition}{Condition}{Conditions}

\crefname{question}{question}{questions}
\Crefname{question}{Question}{Questions}

\crefname{remark}{remark}{remarks}
\Crefname{remark}{Remark}{Remarks}

\crefname{remarks}{remarks}{remarks}
\Crefname{remarks}{Remarks}{Remarks}

\crefname{aside}{aside}{asides}
\Crefname{aside}{Aside}{Asides}

\crefname{note}{note}{notes}
\Crefname{note}{Note}{Notes}

\crefname{appendix}{appendix}{appendices}
\Crefname{appendix}{Appendix}{Appendices}

\newcommand{\secref}[1]{\hyperref[#1]{\S\ref*{#1}}}

\definecolor{backcolour}{rgb}{0.63, 0.79, 0.95}

\lstdefinestyle{mystyle}{
  backgroundcolor=\color{backcolour},
  basicstyle=\ttfamily\footnotesize,
  breakatwhitespace=false,
  breaklines=true,
  captionpos=b,
  keepspaces=true,
  numbers=left,
  numbersep=5pt,
  showspaces=false,
  showstringspaces=false,
  showtabs=false,
  tabsize=2
}

\begin{document}
\author{Mark Whitmeyer \and Cole Williams\thanks{MW: Arizona State University, \href{mailto:mark.whitmeyer@gmail.com}{mark.whitmeyer@gmail.com}. CW:  University of Nebraska-Lincoln \href{mailto:cole.randall.williams@gmail.com}{cole.randall.williams@gmail.com}. We thank the editor (Ben Brooks), three anonymous reviewers, and Emir Kamenica for their feedback.}}
\title{Strong Dominance for Dynamic Signals}
\maketitle

\begin{abstract}
In this paper, we reveal that the signal representation of information introduced by Gentzkow and Kamenica (2017) can be applied profitably to dynamic decision problems. We use this to characterize when one dynamic information structure is more valuable to an agent than another, irrespective of what other dynamic sources of information the agent may possess. Notably, this robust dominance is equivalent to an intuitive dynamic version of Brooks, Frankel, and Kamenica (2022)'s \textit{reveal-or-refine} condition.
\end{abstract}

\section{Introduction}

An agent faces a dynamic decision problem. There is some (possibly infinite) time horizon, and every period the agent takes an action. The possible merits and deficiencies of actions are unknown and depend on a random state of the world. This could be an investor choosing how to dynamically adjust her portfolio, an entrepreneur choosing when to launch (or quit) a startup, or a cervid choosing each day whether, when, and where to forage. 

In all such problems, information is valuable. As it could always be disregarded, a decision-maker can never be hurt by learning more about the world. Naturally, some kinds of information are more valuable than other varieties, and \cite{greenshtein1996comparison} characterizes precisely how to compare two dynamic sources of information. But in this classical notion, each information source being compared is assumed to be the agent's only source of information. This inspires the question, then, of how to compare dynamic structures in the presence of additional sources of information.

To fix ideas, let us modify the newspaper subscription metaphor introduced by \cite*{brooks2022comparisons} (henceforth \textcolor{NavyBlue}{BFK}). An agent is considering year-long subscriptions to either the \textit{Financial Times} or the \textit{Wall Street Journal}. Every day over the course of her subscription the agent will read the newspaper before taking various actions--she faces a fundamentally dynamic problem and it is important to evaluate the newspapers as dynamic sources of information. Moreover, the agent also has additional sources of information--intermittent phone calls and texts from relatives, messages from friends, and a cable TV subscription. In the presence of these additional references, can we nevertheless rank the two specified newspapers by virtue of their value to the agent?

In this paper we formulate a robust notion of sequential dominance, extending the static concept formulated by \textcolor{NavyBlue}{BFK} to a dynamic setting. One dynamic information structure strongly dominates another if for any decision problem and any auxiliary sources of information an agent always prefers the one to the other. \textcolor{NavyBlue}{BFK} show that the static version of this strong dominance is equal to an intuitive ``reveal-or-refine'' comparison between the two information sources under comparison. In short, in the static setting the strongly dominant information structure is such that the agent always either knows ``what she would have seen'' in the dominated information structure or knows what the state is.

Our main result is that in dynamic decision problems, strong sequential dominance is equivalent to period-by-period static strong dominance. The strongly dominant dynamic information structure must be such that \textbf{every period} it ``reveal-or-refines'' the dominated one. Thanks to \textcolor{NavyBlue}{BFK}, necessity is the easiest direction of the proof: if at any period \(t\) the reveal-or-refine condition does not hold, we can just specify a decision problem for which only the agent's decision at period \(t\) matters and so the static necessity finding of \textcolor{NavyBlue}{BFK} applies. Sufficiency is also easy: dominance is essentially replicability, and reveal-or-refine at each period means the more informed agent either always knows what the less informed one does (or better) and hence can mimic any strategy of the latter.

Our finding is valuable for several reasons. First, it has intrinsic worth: asking for a ranking of dynamic information structures that is robust to other sources is obviously relevant--indeed, that paradigm is surely the modal one in the real world. Second, the equivalent (dynamic) ``reveal-or-refine'' condition is simple and intuitive. Third, the ranking of dynamic information structures (absent robustness) is strikingly ungainly. In contrast to the static scenario, where the ``garbling'' equivalence of dominance yields an easy-to-check condition on stochastic matrices, the ``dynamic garbling'' equivalence of sequential dominance is relatively impenetrable and cannot be distilled into a collection of static conditions. In contrast, as we reveal, the conditions for strong sequential dominance are just an aggregation of static conditions and the visual apparatus highlighted by \textcolor{NavyBlue}{BFK} makes dominance easy to comprehend.

Though we feel that our work has value, we do not wish to overstate our contribution. The crucial novelty is our observation that the ``signal'' apparatus of \cite{gentzkow2017bayesian} and \textcolor{NavyBlue}{BFK} is extremely useful in dynamic decision problems. Given this, our main result falls out essentially without effort and is almost an instantiation of Theorem 1 in \textcolor{NavyBlue}{BFK}.

\subsection{Other Related Work}

Though the seminal work ranking static information structures by Blackwell \citep{blackwell,blackwell2} has been around for a long time, the analogous ranking for dynamic information structures constructed by \cite{greenshtein1996comparison} is surprisingly recent. \cite{de2018blackwell} connects the results in the static and dynamic settings in an elegant way and proffers startlingly graceful proofs.

\cite{de2023rationalizing} leverage the apparatus of \cite{greenshtein1996comparison} and \cite{de2018blackwell} to characterize what sequences of actions in a dynamic decision problem can be rationalized by some dynamic information structure. \cite{deb2021dynamic} conduct a similar exercise but with the agent's preferences as a free parameter as well. \cite{renou2024comparing} formulate a dominance notion of dynamic information structures in which the dominant one must be preferred to the other, at \textit{every} history, not just from the \textit{ex ante} perspective.

In a similar spirit to \textcolor{NavyBlue}{BFK}, \cite*{borgers2013signals} study which signals are complements and substitutes. \cite*{de2023robust} ask how to robustly aggregate potentially correlated sources of information. \cite*{liang2022dynamically} look at dynamic information acquisition of information about a (multivariate-)normally distributed state by observing diffusion processes.

Finally, the signal apparatus introduced by \cite*{gentzkow2017bayesian}--which pertains to static environments--is central to our results, as it is the correct representation of information for facilitating robust comparisons and is easily extended (as we reveal) to dynamic settings. The signal construction is related to the object introduced by \cite{green2022two} (originally in their 1978 working paper) and is later used in \cite*{frankel2019quantifying} and \cite*{brooks2022information} (and, of course, in \textcolor{NavyBlue}{BFK}).

\section{Analysis}

\paragraph{Model.}

Our model extends the signal concept introduced by \cite{gentzkow2017bayesian} to a dynamic setting, in which information arrives over a discrete sequence of time periods,  \(t \in \left\{1,\dots,T\right\} \eqqcolon \mathcal{T}\), with \(1 \leq T \leq \infty\). 

There is a finite set of possible states of the world \(\Theta\). The realized state is drawn from some (full-support) prior \(\mu \in \inter \Delta\left(\Theta\right)\). Given a collection of finite sets of signal realizations \(\mathcal{S} \coloneqq \left\{S_1, \dots, S_T\right\}\), a \textit{dynamic experiment} is a stochastic map \(\pi \colon \Theta \to \Delta (S)\), where \(S \coloneqq \times_{t=1}^T S_t\). Each period \(t\), the agent observes a realization \(s_t \in S_t\) according to \(\pi\).

Following \cite{gentzkow2017bayesian}, we define a \textit{signal} to be a finite partition of \(\Theta \times \left[0,1\right]\) for which each element is a Lebesgue measurable subset of \(\Theta \times \left[0,1\right]\). Signal \(\alpha\) \textit{refines} signal \(\beta\) if every element of \(\alpha\) is a subset of an element of \(\beta\). Given two signals, \(\alpha\) and \(\beta\), we define \(\alpha \vee \beta\) to be the coarsest refinement of both \(\alpha\) and \(\beta\). \(\vee\) is the \textit{join} and \(\alpha \vee \beta\) is the signal corresponding to the observation of both \(\alpha\) and \(\beta\).

\begin{definition}
    A \emph{dynamic signal} is a sequence of signals, \(\eta = \left(\eta_{t}\right)_{t=1}^{T}\), such that \(\eta_{t+1}\) refines \(\eta_{t}\) for all \(t < T\).
\end{definition}

For dynamic signals \(\alpha = \left(\alpha_{t}\right)_{t=1}^{T}\) and \(\beta = \left(\beta_{t}\right)_{t=1}^{T}\), we let \(\alpha \vee \beta \coloneqq \left(\alpha_t \vee \beta_{t}\right)_{t=1}^{T}\). An \textit{extended dynamic decision problem} \(\mathcal{D} \coloneqq \left(u, A, \rho\right)\) consists of a collection of compact period-\(t\) action sets \(A \coloneqq \times_{t=1}^T A_t\), 
a continuous utility function \(u \colon A \times \Theta \to \mathbb{R}\), and a dynamic signal \(\rho = \left(\rho_t\right)_{t=1}^T\). Given two dynamic signals \(\alpha\) and \(\beta\), let \(\pi_{\alpha \vee \beta}\colon \Theta \to \Delta\left(S_{\alpha \vee \beta}\right)\) be the dynamic experiment corresponding to the dynamic signal \(\alpha \vee \beta\), where \(S_{\alpha \vee \beta} \coloneqq \times_{t=1}^T \left(S_{\alpha, t} \times S_{\beta, t}\right)\).

\begin{example}\label{ex1}
Suppose there are two periods (\(T=2\)) and two states (\(\Theta = \left\{\theta_L,\theta_H\right\}\)). Figure \ref{fig1} illustrates a dynamic signal \(\eta\), where
\[\eta_1 = \left\{\textcolor{MidnightBlue}{\underbrace{\left(\theta_L, \left[0,\frac{1}{4}\right)\right)}_{h}}, \textcolor{OrangeRed}{\underbrace{\left(\theta_L, \left[\frac{1}{4}, 1\right]\right)}_{l}}, \textcolor{MidnightBlue}{\underbrace{\left(\theta_H, \left[0,\frac{3}{4}\right)\right)}_{h}}, \textcolor{OrangeRed}{\underbrace{\left(\theta_H, \left[\frac{3}{4}, 1\right]\right)}_{l}}\right\}\text{,}\]
and
\[\eta_2 = \left\{\textcolor{MidnightBlue}{\underbrace{\left(\theta_L, \left[0,\frac{1}{4}\right)\right)}_{h H}}, \textcolor{Apricot}{\underbrace{\left(\theta_L, \left[\frac{1}{4}, \frac{3}{4}\right)\right)}_{l H}}, \textcolor{Emerald}{\underbrace{\left(\theta_L, \left[\frac{3}{4}, 1\right]\right)}_{l L}}, \textcolor{MidnightBlue}{\underbrace{\left(\theta_H, \left[0,\frac{3}{4}\right)\right)}_{h H}}, \textcolor{Emerald}{\underbrace{\left(\theta_H, \left[\frac{3}{4}, 1\right]\right)}_{l L}}\right\}\text{.}\]
Then, the corresponding experiment is, letting \(S_1 \coloneqq \left(l,h\right) \quad \text{and} \quad S_2 \coloneqq \left(L,H\right)\),
\[\begin{split}
&\textcolor{MidnightBlue}{\pi\left(h,\left.H\right|\theta_L\right) = \frac{1}{4}, \quad \pi\left(h,\left.H\right|\theta_H\right) = \frac{3}{4}}, \quad \textcolor{Apricot}{\pi\left(l,\left.H\right|\theta_L\right) = \frac{1}{2}, \quad \pi\left(l,\left.H\right|\theta_H\right) = 0}\\
&\textcolor{Emerald}{\pi\left(l,\left.L\right|\theta_L\right) = \frac{1}{4}, \quad \pi\left(l,\left.L\right|\theta_H\right) = \frac{1}{4}}, \quad \text{and} \quad \pi\left(h,\left.L\right|\theta_L\right) = \pi\left(h,\left.L\right|\theta_H\right) = 0\text{.}
\end{split}\]

\begin{figure}
    \centering
    \includegraphics[width=\textwidth]{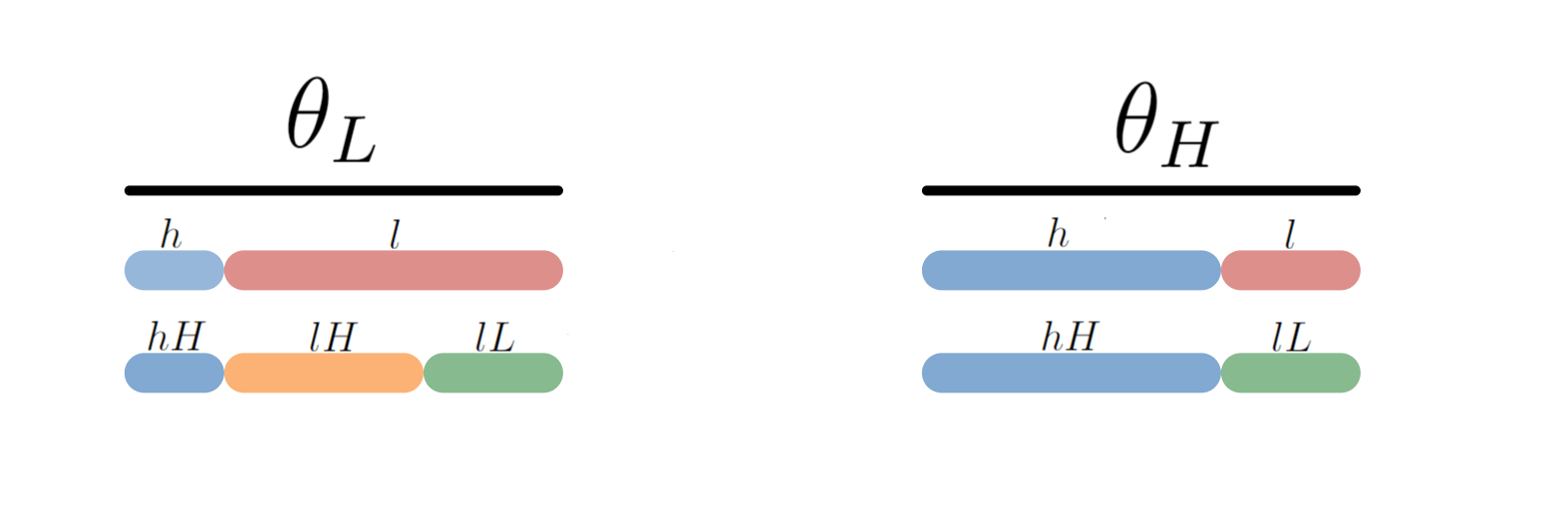}
    \caption{Dynamic signal \(\eta = \left(\eta_1,\eta_2\right)\) from Example \ref{ex1}.}
    \label{fig1}
\end{figure}
    
\end{example}

Let \(a(s)\) denote a deterministic (which is without loss of generality) adapted map \(a \colon S \to A\), where \(a\) is \textit{adapted} if, for each \(t\), the projection of \(a(s)\) onto \(A^t \coloneqq A_1 \times \dots \times A_t\) depends only on \((s_1,\dots,s_t) \in S^{t} \coloneqq S_1 \times \dots \times S_t\) \citep*{de2018blackwell}. That is, the agent is not clairvoyant. Let \(\mathcal{A}(\pi)\) denote the set of adapted maps. Thus, the value of dynamic signal \(\eta\) in an extended dynamic decision problem \(\mathcal{D}\) is:\footnote{Endow all product spaces with their product topologies and, when \(T=\infty\), read sums over \(S\) as integrals under the probability measure induced by the prior and the relevant experiment. Adapted maps form a product of compact action sets indexed by dates and finite histories, so the set of adapted maps is compact by Tychonoff’s theorem. As \(u\) is continuous, for each fixed experiment the expected payoff is continuous on this compact set and so maximum in \(W(\eta)\) is attained.}
\[W(\eta) = \max_{a(s) \in \mathcal{A}(\pi_{\eta \vee \rho})} \sum_{\substack{\theta \in \Theta\\ s \in S}} \pi_{\eta \vee \rho}\left(\left.s\right|\theta\right)\mu(\theta)u\left(a(s),\theta\right)\text{.}\]

\begin{definition}
    Dynamic signal \(\eta\) \emph{strongly dominates} dynamic signal \(\hat{\eta}\) if  \(W\left(\eta\right) \geq W\left(\hat{\eta}\right)\) for any extended dynamic decision problem \(\mathcal{D}\). 
\end{definition}
For a dynamic signal \(\eta\) and for all \(t \in \mathcal{T}\), \(s^t\) denotes a generic element of partition \(\eta_t\).
\begin{definition}
    Signal \(\alpha\) \emph{reveal-or-refines} signal \(\hat{\alpha}\) if for every \(s \in \alpha\), either
    \begin{enumerate}
        \item \(s\) reveals the state: \(\mathbb{P}\left(\left.s\right|\theta\right)>0\) for at most one \(\theta \in \Theta\); or
        \item \(s \subseteq \hat{s}\) for some \(\hat{s} \in \hat{\alpha}\).
    \end{enumerate}
    Dynamic signal \(\eta\) \textit{dynamically reveal-or-refines} dynamic signal \(\hat{\eta}\) if \(\eta_t\) reveal-or-refines \(\hat{\eta}_t\) for every \(t \in \mathcal{T}\).
\end{definition}

\begin{theorem}\label{maintheorem}
    Dynamic signal \(\eta\) strongly dominates dynamic signal \(\hat{\eta}\) if and only if  \(\eta\) dynamically reveal-or-refines \(\hat{\eta}\).
\end{theorem}
\begin{proof}
    The proof is a simple adaptation of the proof of Theorem 1 in \textcolor{NavyBlue}{BFK}. 
    
    \medskip

    \noindent \(\left(\Leftarrow\right)\) Fix an extended dynamic decision problem. It suffices to show that \(\eta\) is more valuable than \(\hat{\eta}\) conditional on any sequence \(\left(s^t\right)_{t=1}^{T}\) for which \(s^T \subseteq \dots \subseteq s^1\) and \(s^t \in \eta_t\) for every \(t \in \mathcal{T}\).
    There are two possibilities: i. either there exists a \(t^* \in \mathcal{T}\) such that  \(s^t\) reveals the state for all \(t \geq t^*\); or ii. \(s^t\) does not reveal the state  for all \(t \in \mathcal{T}\). In the first case, without loss of generality we define \[t^* \coloneqq \min\left\{t \in \mathcal{T} \colon s^t \ \text{reveals the state} \ \right\}\text{.}\]
    We define \(t^{*} = \infty\) in the second case.

    In both cases, for all \(t < t^{*}\), each \(s^t\) refines each \(\hat{s}^{t}\), i.e., \(s^t \subseteq \hat{s}^{t}\) for all \(t \in \left\{1, \dots, t^{*}-1\right\}\). Accordingly for any sequence \(\left(z^t\right)_{t=1}^T \in \left(\rho_t\right)_{t=1}^{T}\), \(s^t \cap z^t \subseteq \hat{s}^{t} \cap z^t\) for all \(t \in \left\{1, \dots, t^{*}-1\right\}\). The result for the second case, therefore, immediately follows. Moreover, in the first case, conditional on the sequence of realizations \(\left(s^t\right)_{t=1}^{t^*-1}\), nothing can be more valuable than \(\eta\).

    \medskip

    \noindent \(\left(\Rightarrow\right)\) Suppose for the sake of contraposition that \(\eta\) does not dynamically reveal-or-refine \(\hat{\eta}\). This means there exists some \(t \in \mathcal{T}\) at which \(\eta_t\) does not reveal-or-refine \(\hat{\eta}_t\). As there exists an extended dynamic decision problem in which \(u\) depends only on the action taken in period \(t\), Theorem 1 in \textcolor{NavyBlue}{BFK} implies the conclusion. \end{proof}

Reveal-or-refine means that it is as if the more-(strongly) informed agent has seen what the less-informed one did, plus possibly more, or the more-informed agent actually knows the state. Dynamic reveal-or-refine means that at each period \(t\), the more-informed agent has seen what the less-informed one has, plus possibly more, or the more-informed agent has learned the state at some point in the past. Consequently, it is clear that the more informed agent must be better off: information is a good with free disposal so the agent could always ignore the extra information. Why is reveal-or-refine at each moment necessary? This is due to the richness of the set of dynamic decision problems, which includes the set of problems in which it is only the agent's action in a specific period that matters, allowing us to apply Theorem 1 in \textcolor{NavyBlue}{BFK}.

\subsection{Additively-Separable Problems}
\begin{definition}
    An extended dynamic decision problem is \emph{additively separable} or \emph{in the AS class} if the agent's utility function \(u \colon A \times \Theta \to \mathbb{R}\) can be written
    \(u(a,\theta) = \sum_{t=1}^T u_t(a_t,\theta)\), where for each \(t\) \(u_t\) is continuous and the period-\(t\) utilities are such that the sum is finite.
\end{definition}

For the sequential dominance of \cite{greenshtein1996comparison}, in which the agent has no auxiliary sources of information--akin to the environment of \cite{blackwell} and \cite{blackwell2}-- \cite{markcole} show that dominance in the AS class is not equivalent to dominance in the grand class of all sequential decision problems. We find that this non-equivalence \textbf{does not} persist when we modify dominance to strong dominance. Namely, we discover an equivalence of the two concepts.
\begin{definition}
    Dynamic signal \(\eta\) \emph{strongly dominates} dynamic signal \(\hat{\eta}\) in the AS class if for any extended dynamic AS problem \(W\left(\eta\right) \geq W\left(\hat{\eta}\right)\). 
\end{definition}
Then, we have an easy corollary of Theorem \ref{maintheorem}:
\begin{corollary}
    Dynamic signal \(\eta\) strongly dominates dynamic signal \(\hat{\eta}\) in the AS class if and only if \(\eta\) dynamically reveal-or-refines \(\hat{\eta}\).
\end{corollary}
\begin{proof}
\noindent \(\left(\Leftarrow\right)\) Theorem \ref{maintheorem} implies this direction.

\medskip

\noindent \(\left(\Rightarrow\right)\) The construction in the proof of Theorem \ref{maintheorem} is an additively-separable extended dynamic decision problem, so the same argument via contraposition works. \end{proof}

\subsection{(Non-Robust) Sequential Dominance}

As noted above, the sequential dominance of \cite{greenshtein1996comparison} (and \cite{de2018blackwell}) pertains to a scenario in which the information structures being compared are the agent's only source of information. Such a \textit{dynamic decision problem} \(\mathcal{\Tilde{D}} \coloneqq \left(u, A\right)\) consists of a collection of compact period-\(t\) action sets \(A \coloneqq \times_{t=1}^T A_t\) and a continuous utility function \(u \colon A \times \Theta \to \mathbb{R}\). Letting \(\pi_{\eta}\) denote the dynamic experiment corresponding to dynamic signal \(\eta\), and recalling that  \(\mathcal{A}(\pi)\) denotes the set of adapted maps given experiment \(\pi\), the value of dynamic signal \(\eta\) in \(\mathcal{\tilde{D}}\) is
\[\tilde{W}(\eta) = \max_{a(s) \in \mathcal{A}(\pi_{\eta})} \sum_{\substack{\theta \in \Theta\\ s \in S}} \pi_{\eta}\left(\left.s\right|\theta\right)\mu(\theta) u\left(a(s),\theta\right)\text{.}\]
\begin{definition}
    Dynamic signal \(\eta\) \emph{dominates} dynamic signal \(\hat{\eta}\) if  \(\tilde{W}\left(\eta\right) \geq \tilde{W}\left(\hat{\eta}\right)\) for any dynamic decision problem \(\mathcal{\tilde{D}}\). 
\end{definition} 
A corollary of Theorem \ref{maintheorem} is, therefore,
\begin{corollary}
    Dynamic signal \(\eta\) dominates dynamic signal \(\hat{\eta}\) if \(\eta\) dynamically reveal-or-refines \(\hat{\eta}\).
\end{corollary}
The converse of this result is false, however. Indeed, that a static experiment can dominate (in the Blackwell sense) another despite its corresponding signal not dominating the other's is noted by \cite{gentzkow2017bayesian}.

\subsection{More, Sooner}

A counterintuitive phenomenon exposed by \cite{greenshtein1996comparison} is that a decision maker may strictly prefer to observe a less informative experiment before a more informative one. This is illustrated starkly in \citet[Example 2.1]{greenshtein1996comparison}, which we now adapt.

There are two states, \(\Theta=\{\theta_L,\theta_R\}\), and two periods. The prior is \(\mu(\theta_R)=4/5\). There are two conditionally-independent binary statistical experiments \(\pi_{\alpha} \colon \Theta \to \Delta\left(\left\{\ell_{\alpha},r_{\alpha}\right\}\right)\) and \(\pi_{\beta} \colon \Theta \to \Delta\left(\left\{\ell_{\beta},r_{\beta}\right\}\right)\), with
\[
\pi_{\alpha}\left(r_{\alpha}\mid \theta_R\right) = \frac{3}{4} = \pi_{\alpha}\left(\ell_{\alpha}\mid \theta_L\right),\qquad \text{and} \qquad \pi_{\beta}\left(r_{\beta}\mid \theta_R\right) = \frac{2}{3} = \pi_{\beta}\left(\ell_{\beta}\mid \theta_L\right),
\]
so that experiment \(\pi_{\alpha}\) strictly Blackwell dominates \(\pi_{\beta}\). Let \(\alpha\) and \(\beta\) denote a signal representation of this conditionally-independent pair, so that, within each state, the \(\alpha\)- and \(\beta\)-cells have the above marginal probabilities and their intersections have the product probabilities.\footnote{Concretely, for each \(\theta\), partition \(\{\theta\}\times[0,1]\) into four cells indexed by \((s_\alpha,s_\beta)\in\{\ell_\alpha,r_\alpha\}\times\{\ell_\beta,r_\beta\}\), with Lebesgue measures \(\pi_\alpha(s_\alpha\mid\theta)\pi_\beta(s_\beta\mid\theta)\). Let \(\alpha\) group these cells according to \(s_\alpha\), and let \(\beta\) group them according to \(s_\beta\).}

The decision problem faced by the DM is as follows. In period one, she has three actions \(\left\{D,L,R\right\}\) (defer, left, or right); and in the second period, she has two \(\left\{L,R\right\}\). Letting \(c>0\), her state-dependent payoff is, for any second-period action \(a\in\{L,R\}\),
\[
\begin{array}{c|cccc}
&\left(L,a\right)&\left(R,a\right)&\left(D,L\right)&\left(D,R\right)\\
\hline
\theta_L&1&0&1-c&-c\\
\theta_R&9/10&1&9/10-c&1-c
\end{array}
\]
This is a simple Wald-style sequential sampling problem \citep{wald1947sequential}: after the DM observes an experiment in period \(1\), she can either ``act then'' by choosing \(L\) or \(R\), or ``defer'' the decision to period \(2\), incurring a cost \(c > 0\) but enabling herself to take advantage of the further information she receives from the period-\(2\) experiment.\footnote{That is, in period \(1\) she either takes a decision or pays a cost to obtain more information (the period-\(2\)) experiment before making her decision.} We further specify that \(0<c<4/195\).

We compare two dynamic experiments: in the first, the DM observes \(\pi_{\alpha}\) in the first period, then \(\pi_{\beta}\) in the second (\((\pi_{\alpha}, \pi_{\beta})\)); whereas the second flips the order of observation (\((\pi_{\beta}, \pi_{\alpha})\)). As we now argue, in this particular decision problem, given her particular prior, the DM prefers \((\pi_{\beta}, \pi_{\alpha})\) to \((\pi_{\alpha}, \pi_{\beta})\), \textit{viz.}, she prefers to observe the less-informative experiment first. 

When the DM is choosing between \(L\) and \(R\), she prefers \(R\) if and only if her belief that the state is \(R\) is sufficiently high; specifically, if \(q \coloneqq \mathbb{P}(\theta_R) \geq 10/11\). Moreover, the posteriors after observing both experiments are \((2/5, 8/11, 6/7, 24/25)\), so it is only after observing both \(r_{\alpha}\) and \(r_{\beta}\) that the DM prefers \(R\) to \(L\). Consequently, no matter the sequence of experiments, if the DM observes the left realization (either \(\ell_{\alpha}\) or \(\ell_{\beta}\)) in the first period, she prefers to act then (and take action \(L\)) as the second-period experiment will have no value but impose on her the cost \(c\).

On the other hand, no matter the sequence of experiments, the DM prefers \(D\) if the first signal realization is right (either \(r_{\alpha}\) or \(r_{\beta}\)). In both cases, the second experiment now is strictly valuable, and this value exceeds its cost (provided \(c < 4/195\)).\footnote{If the dynamic experiment is \((\pi_{\beta}, \pi_{\alpha})\), the DM's posterior following \(r_{\beta}\) is \(8/9\). Given this, her probability of observing \(r_{\alpha}\) is 25/36. Thus, her payoff from choosing \(D\) then behaving optimally (\(R\) in period \(2\) if and only if \(r_{\alpha}\)) is, writing \(v(x) \coloneqq \max\left\{1 - x/10,x\right\}\), \[
\frac{25}{36}v\left(\frac{24}{25}\right)+\frac{11}{36}v\left(\frac{8}{11}\right) - c=\frac{19}{20} -c,
\] whereas her payoff from ``acting then" and choosing \(L\) is \(v(8/9)=41/45\). Consequently, she strictly prefers to ``defer'' (\(D\)) if and only if \(c < 7/180\). The same logic holds when the sequential experiment is \((\pi_{\alpha}, \pi_{\beta})\), in which case the cost threshold is \(4/195\).} 

Now let us finish comparing the two dynamic experiments. Observe that under both, the DM's ultimate choice between \(R\) and \(L\) is identical: she picks \(L\) unless she observes both \(r_{\alpha}\) and \(r_{\beta}\). Consequently, her payoffs, \textit{gross of costs}, are identical.\footnote{This payoff is \((4/5)\left((1/2)\cdot1+(1/2)\cdot(9/10)\right)+(1/5)\cdot(11/12)=283/300\).} The only difference between the two dynamic experiments is the probability of paying the waiting cost: when the DM sees \(\pi_{\alpha}\) first, this probability is \(\mathbb{P}_{\pi_{\alpha}}(r_{\alpha}) = 13/20\); which is greater than the cost probability under \(\pi_{\beta}\), \(\mathbb{P}_{\pi_{\beta}}(r_{\beta}) = 3/5\). As a result, the DM strictly prefers to observe \(\pi_{\beta}\), the strictly less-informative experiment, first.

Notably, in contrast, strong dominance does not produce this apparent pathology. The following corollary is an immediate consequence of Theorem \ref{maintheorem}.

\begin{corollary}\label{cor:more_sooner}
For dynamic signals \(\eta\coloneqq(\alpha,\alpha\vee\beta)\) and \(\eta'\coloneqq(\beta,\alpha\vee\beta)\), \(\eta\) strongly dominates \(\eta'\) if and only if \(\alpha\) strongly dominates \(\beta\).
\end{corollary}

The example in this section illustrates why the hypothesis in Corollary \ref{cor:more_sooner} is stronger than Blackwell dominance. Although the experiment corresponding to \(\alpha\) (\(\pi_\alpha\)) strictly Blackwell-dominates the experiment corresponding to \(\beta\) (\(\pi_\beta\)), \(\alpha\) does not strongly dominate \(\beta\). By Theorem \ref{maintheorem}, strong dominance is equivalent to reveal-or-refine. But in the example, no \(\alpha\)-cell either reveals the state or is contained in a \(\beta\)-cell, so \(\alpha\) does not reveal-or-refine \(\beta\).

The reason is precisely the conditional independence of the two experiments. For every \(s_\alpha\in\{\ell_\alpha,r_\alpha\}\) and every \(\theta\in\{\theta_L,\theta_R\}\),
\[
\mathbb{P}(\beta=r_\beta\mid \alpha=s_\alpha,\theta)
=
\mathbb{P}(\beta=r_\beta\mid \theta)\in(0,1).
\]
Thus, after any realization of \(\alpha\), both realizations of \(\beta\) remain possible in each state; an \(\alpha\)-cell cannot be contained in a \(\beta\)-cell unless \(\beta\) is conditionally degenerate. More generally, if \(\alpha\) reveal-or-refines \(\beta\) and an \(\alpha\)-cell does not reveal the state, then observing that \(\alpha\)-cell already determines the \(\beta\)-cell. If \(\alpha\) and \(\beta\) were also conditionally independent given the state, \(\beta\) would have to be conditionally degenerate on the relevant states. The experiments above are not conditionally degenerate, so the conditionally-independent structure that drives the reversal cannot be reveal-or-refine ordered.

This example highlights the contrast between Blackwell dominance and strong dominance. Blackwell dominance does not say that the actual realization of \(\beta\) is known after observing \(\alpha\). Reveal-or-refine imposes exactly this stronger history-by-history relation, except on histories at which the state has already been revealed. Consequently, under strong dominance the \(\alpha\)-first agent can mimic any strategy available to the \(\beta\)-first agent on non-state-revealing histories, including any act-now versus defer decision. The fresh residual information that drives the example is, therefore, absent.

\subsection{Interim Versus Ex Ante Comparisons}\label{subsec:interim-vs-exante}

\textcolor{NavyBlue}{BFK} establish that in the static setting, their main \textit{ex ante} notion is equivalent to an interim notion of strong dominance. We now show that the analogous equivalence persists in dynamic environments. Consider an extended dynamic decision problem $\mathcal D=(u,A,\rho)$. Let \(\eta\) be a dynamic signal. For $t\in\mathcal T$, $s^t\in\eta_t$, and $z^t\in\rho_t$ with $\mathbb P(s^t\cap z^t)>0$, we define the interim value of dynamic signal \(\eta\) to be
\[
W(\eta \mid s^t,z^t) \coloneqq
\max_{a(s)\in\mathcal A(\pi_{\eta \vee\rho})}
 \mathbb E\left[u \left(a(s),\theta\right)\,\middle|\, s^t\cap z^t\right],
\]
where the conditional expectation is taken with respect to the distribution induced by the prior $\mu$ and the dynamic experiment $\pi_{\eta\vee\rho}$. For any other dynamic signal \(\hat{\eta}\), define
\[
W(\hat \eta \mid s^t,z^t) \coloneqq
\max_{a(s)\in\mathcal A(\pi_{\hat \eta \vee\rho})}
 \mathbb E\left[u \left(a(s),\theta\right)\,\middle|\, s^t\cap z^t\right],
\]
where it is important to note that this is conditional on the same event \(s^t\cap z^t\) as \(W(\eta \mid s^t,z^t)\).

This mirrors \textcolor{NavyBlue}{BFK}'s interim comparison: we fix an event of the auxiliary information and the more refined signal and compare the continuation values of different signals conditional on that event.

\begin{proposition}\label{prop:interim-exante}
For dynamic signals $\eta$ and $\hat\eta$, the following are equivalent:
\begin{enumerate}
\item\label{interim1} $\eta$ strongly dominates $\hat\eta$.
\item\label{interim2} For any extended dynamic decision problem $\mathcal D=(u,A,\rho)$, any $t\in\mathcal T$, and any
$s^t\in\eta_t$, $z^t\in\rho_t$ with $\mathbb P(s^t\cap z^t)>0$, we have \(W(\eta\mid s^t,z^t) \geq W(\hat\eta\mid s^t,z^t)\).
\end{enumerate}
\end{proposition}

\begin{proof}
\noindent (\ref{interim2} $\Rightarrow$ \ref{interim1})
Fix an extended dynamic decision problem $\mathcal D=(u,A,\rho)$. Let $\hat a^*(s)\in\mathcal A(\pi_{\hat\eta\vee\rho})$ attain $W(\hat\eta)$ in $\mathcal D$. From \ref{interim2} with $t=1$, for every $s^1\in\eta_1$ and $z^1\in\rho_1$ with $\mathbb P(s^1\cap z^1)>0$,
\[
W(\eta\mid s^1,z^1) \geq
\mathbb E \left[u \left(\hat a^*(s),\theta\right)\,\middle|\, s^1\cap z^1\right].
\]

Now take expectations across the partition $\eta_1\vee\rho_1$.
By the usual ``gluing'' argument permitted by adaptedness, the \textit{ex ante} value is the average
of the conditional optima:
\[
W(\eta) = \sum_{c\in \eta_1\vee\rho_1}\mathbb P(c) W(\eta\mid c),
\quad\text{where} \quad W(\eta\mid c) \coloneqq \max_{a(s)\in\mathcal A(\pi_{\eta\vee\rho})}\mathbb E[u(a(s),\theta)\mid c].
\]
Indeed, for each $c$ pick an adapted optimizer $a^c(\cdot)$ attaining $W(\eta\mid c)$, then define
$a(\cdot)$ by setting $a(\cdot)=a^c(\cdot)$ on histories whose period-1 realization lies in $c$. 

Therefore,
\[W(\eta)
=\sum_{c\in \eta_1\vee\rho_1}\mathbb P(c) W(\eta\mid c)
\ge \sum_{c\in \eta_1\vee\rho_1}\mathbb P(c) \mathbb E \left[u \left(\hat a^*(s),\theta\right)\,\middle|\, c\right] =\mathbb E \left[u \left(\hat a^*(s),\theta\right)\right]
= W(\hat\eta).\]
Since $\mathcal D$ was arbitrary, $\eta$ strongly dominates $\hat\eta$.

\medskip \noindent (\ref{interim1} $\Rightarrow$ \ref{interim2})
Fix an extended dynamic decision problem $\mathcal D=(u,A,\rho)$. By Theorem~\ref{maintheorem}, $\eta$ dynamically reveal-or-refines $\hat\eta$.

Now fix $t\in\mathcal T$, $s^t\in\eta_t$, and $z^t\in\rho_t$ with $\mathbb P(s^t\cap z^t)>0$.
There are two cases. If $s^t$ reveals the state, then conditional on $s^t\cap z^t$, the state is known under $\eta$.
Consequently, nothing can be more valuable than $\eta$ conditional on $s^t\cap z^t$,\footnote{Because \(s^t\) reveals the state, \(s^t \cap z^t\) pins down the realized state \(\theta^*\), so any constant adapted map \(a(s) = a^* \in \argmax_{a \in A}u(a,\theta^*)\) delivers the DM's maximal payoff conditional on \(s^t \cap z^t\). This is also true under \(\hat{\eta}\), so the stated inequality is an equality; \textit{viz.}, $W(\eta\mid s^t,z^t) = W(\hat\eta\mid s^t,z^t)$.} so $W(\eta\mid s^t,z^t)\ge W(\hat\eta\mid s^t,z^t)$.

Now suppose $s^t$ does not reveal the state. Since $\eta$ dynamically reveal-or-refines $\hat\eta$, for each $\tau\le t$ the realized element
$s^\tau\in\eta_\tau$ satisfies $s^\tau\subseteq \hat s^\tau$
for some $\hat s^\tau\in\hat\eta_\tau$.
Because $\hat\eta_\tau$ is a partition, this $\hat s^\tau$ is unique whenever it exists.
Accordingly, for any continuation of the auxiliary signal $(z^\tau)_{\tau\le t}$ consistent with $z^t$,
we have
\[
s^\tau\cap z^\tau  \subseteq \hat s^\tau\cap z^\tau,
\quad\text{for all }\tau\le t.
\]

Now let $\hat a(\cdot)\in\mathcal A(\pi_{\hat\eta\vee\rho})$ be arbitrary.
We construct $a(\cdot)\in\mathcal A(\pi_{\eta\vee\rho})$ by mimicking $\hat a(\cdot)$ history-by-history:
for each realization of $(\eta\vee\rho)$ and each $\tau\in\mathcal T$, let $\hat s^\tau$ be the (unique)
element of $\hat\eta_\tau$ containing the realized $s^\tau$ whenever $s^\tau$ does not reveal the state,
and define the $\tau$-th component of $a(\cdot)$ to equal the $\tau$-th component of $\hat a(\cdot)$ evaluated
at the corresponding $(\hat s^\tau,z^\tau)$ history. If at some later date $\tau>t$ the realized element $s^\tau$ under $\eta$ reveals the state $\theta$, then from that date onward we may (re)define the remaining components $\{a_k\}_{k\ge \tau}$ so as to maximize the continuation payoff given $\theta$ and the already-chosen past actions $\{a_k\}_{k<\tau}$.\footnote{Since on any element \(s^\tau\) that reveals \(\theta\), a continuation plan that is constant across future histories (but may depend on the revealed \(\theta\) and the already-chosen past actions) is adapted, this modification is feasible and can only weakly increase \(\mathbb E[u(a(s),\theta)\mid s^t\cap z^t]\) relative to continued mimicking.}
 This is well-defined and adapted because, for each $\tau$,
the projection of $a(\cdot)$ onto $A_\tau$ depends only on the period-$\tau$ history.

On the event $s^t\cap z^t$, let \(\hat{s} \in S_{\hat{\eta} \vee \rho}\) denote the realized element under \(\pi_{\hat{\eta} \vee \rho}\) induced by the same realized point in \(\Theta \times [0,1]\). Moreover, the constructed $a(\cdot)$ and the original $\hat a(\cdot)$ induce the same
action history on all continuations before any possible revelation of the state. Then,
\[
\mathbb E \left[u \left(a(s),\theta\right)\,\middle|\, s^t\cap z^t\right]
\ \ge\
\mathbb E \left[u \left(\hat a(\hat s),\theta\right)\,\middle|\, s^t\cap z^t\right].
\]
Since $\hat a(\cdot)$ was arbitrary, taking the maximum over $\hat a(\cdot)$ on the right and then the maximum
over $a(\cdot)$ on the left yields $W(\eta\mid s^t,z^t)\ge W(\hat\eta\mid s^t,z^t)$.\end{proof}


\newpage

\section*{Declarations}

The authors did not receive support from any organization for the submitted work. The authors have no relevant financial or non-financial interests to disclose.

\newpage

\bibliography{sample.bib}

\end{document}